\documentclass[11pt]{article}

\pdfoutput=1  
\usepackage{microtype}
\usepackage{graphicx}
\usepackage{authblk}
\usepackage{amsthm}
\usepackage{color}
\usepackage{url}
\usepackage[numbers]{natbib}
\usepackage{listings}
\usepackage{float}  
\usepackage{booktabs}
\usepackage{textcomp}  
\usepackage[colorlinks,bookmarksopen,bookmarksnumbered,citecolor=red,urlcolor=red]{hyperref}

\newtheorem{theorem}{Theorem}
\newtheorem{lemma}{Lemma}
\newtheorem*{claim}{Claim}

\newif\ifIsSubmission
\IsSubmissionfalse

\begin{document}

\title{Randen - fast backtracking-resistant random generator with AES+Feistel+Reverie}
\author{\href{mailto:janwas@google.com}{J.~Wassenberg}}
\author{R.~Obryk}
\author{J.~Alakuijala}
\author{E.~Mogenet}
\affil{Google Research}
\maketitle

\begin{abstract}
\noindent
Algorithms that rely on a pseudorandom number generator often lose their
performance guarantees when adversaries can predict the behavior of the
generator. To protect non-cryptographic applications against such attacks, we
propose `strong' pseudorandom generators characterized by two properties:
computationally indistinguishable from random and backtracking-resistant. Some
existing cryptographically secure generators also meet these criteria, but
they are too slow to be accepted for general-purpose use. We introduce a new
open-sourced generator called `Randen' and show that it is `strong' in addition
to outperforming Mersenne Twister, PCG, ChaCha8, ISAAC and Philox in real-world
benchmarks.
This is made possible by hardware acceleration. Randen is an instantiation of
Reverie, a recently published robust sponge-like random generator, with a new
permutation built from an improved generalized Feistel structure with 16
branches. We provide new bounds on active s-boxes for up to 24 rounds of this
construction, made possible by a memory-efficient search algorithm. Replacing
existing generators with Randen can protect randomized algorithms such as
reservoir sampling from attack. The permutation may also be useful for
wide-block ciphers and hashing functions.

\end{abstract}

\section{Introduction}

Pseudorandom number generators are very widely used. For example, searching
Github for C++ code containing mt19937 (Mersenne Twister) returns 220,000
hits. Some of these usages will be vulnerable to unexpected correlations
\cite{rngCorrelation} or exploitation by attackers \cite{weakRngDefinition}.
To avoid having to audit each call site, we propose to replace most of them
with our new fast and `strong' generator.

\subsection{Definition of strong}

In this paper, we choose to characterize a strong deterministic random generator
by two properties:
\begin{enumerate}

\item Even relatively powerful adversaries able to generate and store up to
$2^{64}$ random outputs cannot distinguish the output from random unless they
know the current state. This property is useful even for non-cryptographic
applications: it implies empirical randomness, which reduces the likelihood of
flaws such as correlations that might affect simulations \cite{rngCorrelation}.
This property also ensures adversaries cannot predict future outputs, which
makes it harder for them to trigger worst cases in randomized algorithms.

\item Past outputs cannot be reconstructed even after the state is
compromised. This is known as enhanced backward secrecy
\cite{bsiRequirements}, forward security \cite{dodisRequirements} and
backtracking resistance \cite{nistRequirements}. We use the latter name
because it is more clear. This may not be necessary for simulation
applications, but it prevents adversaries from discovering past behavior,
e.g. which inputs were sampled.
\end{enumerate}

\noindent
The notion of `robustness' from the literature also requires generators to
recover security after a state compromise \cite{dodisRequirements}. This is
typically achieved by periodically reseeding from an entropy source. However,
our applications require at least the option of deterministic results for
reproducibility and debugging. Our definition of `strong' describes the
achievable security in this model.

\subsection{Existing generators}

\begin{description}
\item[RC4] was a popular stream cipher designed in 1987, but attacks with
practical complexity have recently been published \cite{rc4Weak}.

\item[ISAAC] was published in 1996 and has some similarities with RC4.
It has many weak initial states, though the resulting biases can be avoided with
a modified algorithm \cite{isaacWeak}. Despite the large 1024 byte state, ISAAC
is relatively fast and widely used \cite{kneuselRNG}[p.~200].

\item[ChaCha20] is an ARX-based hash/stream cipher used by OpenBSD
\texttt{arc4random} and Linux 4.8 \verb|/dev/urandom|. It is an order of
magnitude slower than some general-purpose generators \cite{pcg}[p.~41];
a similar result is observed in our benchmark. Note that a ChaCha20 generator
reportedly fails one part \cite{kneuselRNG}[p.~205] of the dieharder test of
empirical randomness.

\item[Tyche-i] is based on ChaCha and reaches 1.5 cycles per byte
\cite{tycheChacha}. However, its authors discovered some short cycles and
recommended a workaround that is seven times slower \cite{tycheAddon}.

\item[Mersenne Twister] is a popular general-purpose generator included in
the C++ standard library. It is fast but not strong: the generated numbers are
an easily inverted bijection (`tempering function') of a portion of the state,
so adversaries learn the entire state after generating one full buffer.

\item[xorshift128+ and xoroshiro128+] \cite{xoroshiro} fail a PractRand test
due to lack of randomness in the lower bit \cite{practrandXORO}; the
latter is also easily distinguishable from random \cite{xoroCracking}.

\item[Philox] is a noncryptographic counter-based generator with an iterated
bijection based on a Feistel network using double-width multiplications
\cite{salmonRNG}. It passes the TestU01 suite by construction \cite{testu01}.
GPU implementations achieve very high throughput \cite{salmonRNG}, but our
benchmarks indicate our hardware AES-based permutation is twice as fast on CPUs.

\item[PCG] includes an extension that XORs the output of a 128-bit generator
with one of 32 table entries \cite{pcg}. Periodically scrambling the table using
entropy from the state might be sufficient for backtracking resistance, which
the previously mentioned generators lack. However, PCG makes no concrete
security claims \cite{pcgCode}. Although its statistical quality appears
good, we are unaware of any existing proofs of indistinguishability and
backtracking resistance, so PCG is not known to be strong.

\item[AES-CTR] is a block-cipher mode that can be used as a generator by
enciphering all-zero plaintexts \cite{salmonRNG}. Although indistinguishable
from random, this lacks backtracking resistance. Once the attacker knows the
current counter value and key, they can reconstruct all prior outputs.

\item[AES-CTR-DRBG] is a strong generator specified in
NIST~800-90A~\cite{nistRequirements}. Similarly to Fortuna~\cite{fortuna},
it periodically re-keys based on the current output. However, this is about five
times slower than AES-CTR, and too slow for general purpose use (see
Section~\ref{sec:speed}). Note that relaxing the re-keying requirements, e.g.\
only after every 100 blocks, greatly reduces the overhead and could yield a
faster generator. However, applications may not be willing to accept exposing
several thousand prior outputs when the state leaks.

\item[Fast-key-erasure RNGs] are a more efficient alternative to CTR-DRBG
without the potential for unsafe usage \cite{fastErasure}. Bernstein
reiterates the importance of backtracking resistance and proposes to generate
a buffer of random bits using a stream cipher, immediately overwriting its key
with part of the buffer, and returning the rest. However, there are two
integration issues with this approach.  If the stream cipher relies on AVX2 or
AVX-512 SIMD for speed \cite{vectorChacha}, there is a risk of slowing down
the entire application. Frequency throttling has been identified as a cause
\cite{avxThrottle}, but this only applies to ChaCha20-Poly1305. Salsa/ChaCha
are unaffected because they only require low-power operations, whereas the
multiplications in Poly1305 trigger throttling. Instead, we are concerned
about another AVX2 implementation detail: a warmup period of about 60,000
cycles triggered by the
first AVX2 instruction within a 675~\textmu s window. During this time, SIMD
instructions are considerably slower; Haswell CPUs can even stall for
10~\textmu s due to their internal voltage regulator. Thus, sporadic use of
stream ciphers relying on AVX2/AVX-512 can slow down the entire application, or
even unrelated jobs running on the same socket.  The second integration issue
is buffer size. Stream ciphers are considerably slower for small buffers
\cite{ebacsStream}, which are preferred by applications and library writers
because generators may be short-lived or only used to produce a few numbers.
Our proposed approach avoids both issues. First, 128-bit AES hardware runs at
full frequency without warmup, and is performance-portable to other 128-bit
SIMD architectures -- see our measurements in Section~\ref{sec:speed}. Second,
our Feistel permutation does not require a buffer larger than its 256-byte
size.

\end{description}

\noindent
We are unaware of any existing generator that is both strong and fast
in real-world applications.

\subsection{Intended applications}

We argue that the default choice of random generators should be `strong'. This
makes it harder to attack randomized algorithms and trigger skewed samples or
worst-case performance. Security-critical applications such as generating
cryptographic keys should continue to use well-studied and trusted
cryptographic generators such as Fortuna \cite{fortuna}. However, these are
too slow to be accepted for general use. For example, we have tens of thousands
of high-end CPU cores occupied by general-purpose random generators. Thus, the
speed of our proposed generator is important. Because Mersenne Twister is
commonly used in C++ applications (see Github usage above), we assume its
level of performance is generally acceptable. The Randen generator is designed
to reach similar performance.

Note that applications that require many random numbers without any concern
for security (such as Monte-Carlo simulations) may still prefer a faster but
weaker generator such as \texttt{pcg32} \cite{pcgCode}. For other applications,
we suggest using Randen because it is strong and tends to outperform Mersenne
Twister (see Section~\ref{sec:speed}).

\subsection{Contributions}

This paper makes four contributions:
\begin{itemize}
\item Introducing Randen\footnote{RANDen = RANDom number generator, or
beetroots in Swiss German.}, a new generator based on Reverie \cite{reverie}
instantiated with a generalized Feistel structure \cite{improvedFeistel}
(Section~\ref{sec:spec}).

\item Arguing that Randen is `strong', and explaining why this is important
even for non-cryptographic applications (Section~\ref{sec:security}).

\item Showing that existing secure generators are too slow for general
purpose use (Section~\ref{sec:speed}). By contrast, Randen outperforms
Mersenne Twister in some real-world use cases despite providing a higher level
of security. To the best of our knowledge, Randen is the fastest `strong'
software generator.

\item Proposing an efficient algorithm for lower-bounding active s-boxes in
16-branch generalized Feistel networks with SPSP-type round functions
(Appendix~\ref{sec:search}). We provide results for up to 24~rounds, whereas
prior work reaches 18 rounds \cite{feistelSboxes}.
\end{itemize}

\section*{Absence of backdoors}

We, the designers of Randen, faithfully declare that we have not inserted
any weaknesses in this algorithm/implementation, nor have we discovered any
weakness not described in this paper.

\ifIsSubmission\else
\section*{Acknowledgment}

Thanks to Jeffrey Lim, Titus Winters, Chandler Carruth and Daniel Lemire for
suggestions and technical help on improving the benchmarks. We also
appreciate the many clear and accessible posts on practical random
number generation topics by Melissa O'Neill (author of PCG).
\fi

\section{Specification}
\label{sec:spec}

Randen is an instantiation of Reverie, a sponge-like construction that
scrambles its internal state using a permutation \cite{reverie}. To avoid the
ambiguities of pseudocode, we describe its parts using standard C++11,
plus explanatory text.
The permutation operates on 128-bit pieces of the state called `branches'.
This corresponds to the block size of AES. For convenience, we assume the
availability of a platform-specific 128-bit SIMD vector type \texttt{V} with
associated \texttt{Load}, \texttt{Store} and \texttt{AES} functions.

\subsection{Initialization}

Randen operates on a 2048-bit state, of which the first 128 bits are the
inaccessible `inner' portion corresponding to the `capacity' of a sponge. The
remaining `outer' bits are the generated random bits. To simplify
initialization of the state, we partition it into 32 64-bit integers, two per
128-bit branch. Zero-initializing the state yields a valid generator, but
applications will typically set some of its outer bits to arbitrary
user-specified `seed' values. Providing more than 128 seed bits may help
against multi-user attacks involving precomputation. We suggest a 256-bit seed,
specified as four 64-bit \texttt{seed} integers. For more thorough diffusion,
the seeds should be placed into `even-numbered` (according to zero-based
index) branches of the state, e.g.\ the third (with zero-based indices 4, 5 in
the array of 64-bit integers) and fifth.

{\small
\begin{lstlisting}[numbers=none]
  uint64_t state[32];
  memset(state, 0, sizeof(state));
  state[4] = seed0;
  state[5] = seed1;
  state[8] = seed2;
  state[9] = seed3;
\end{lstlisting}}

\subsection{Permutation}

Randen's \texttt{Permute} is a generalized type-2 Feistel network
\cite{type2Feistel} with 16 branches of 128 bits.

\begin{figure}[ht]
\includegraphics[trim=6cm 7.5cm 11cm 7.5cm]{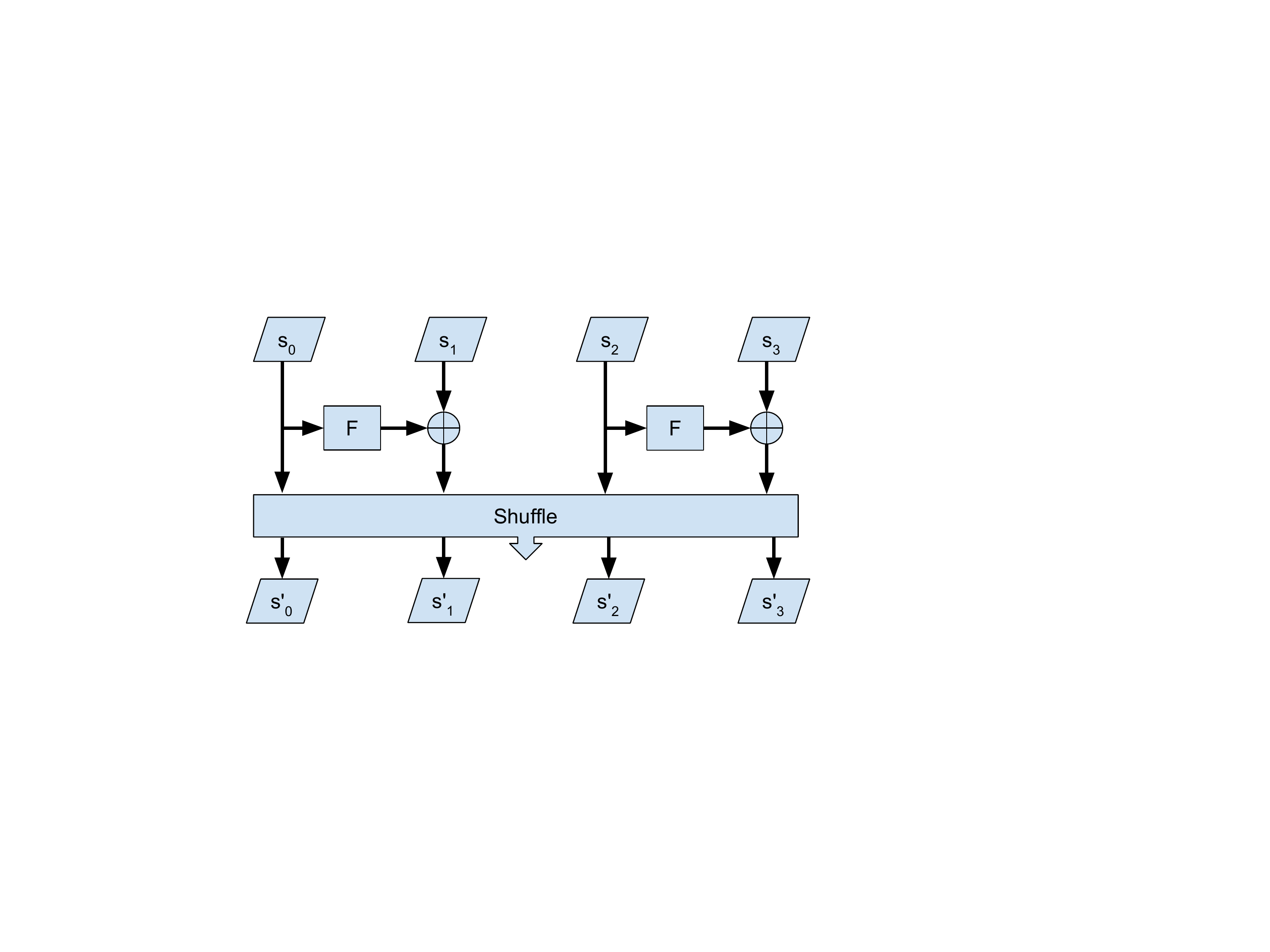}%
\caption{One round of a four-branch type-2 generalized Feistel network
with a block shuffle. F is a 128-bit permutation consisting of two AES
rounds, described below.}
\label{fig:feistel}
\end{figure}

\noindent
It consists of two layers (Figure~\ref{fig:feistel}). The first (denoted
\texttt{RoundFunctions}) XORs odd-numbered branches $\mathrm{s}_\mathrm{i}$
with a function \texttt{F} of their even-numbered neighbors. \texttt{F} is the
same as in Simpira~v2 \cite{simpira}: two rounds of AES. The first round's
constant, denoted \texttt{key}, is unique for every instance of \texttt{F}.
This avoids any potential weaknesses due to weak or structured round
constants, e.g.\ in Simpira v1 \cite{simpira}[p.~12]. We will discuss the size
and purpose of the constants in the description of \texttt{Permute} below.
The second constant is zero, which enables an optimization below.

{\small
\begin{lstlisting}[numbers=none]
  // Round function: two-round AES with a unique round constant.
  V F(const V even, const V key) {
      const V f1 = AES(even, key);
      return AES(f1, zero);
  }
\end{lstlisting}}

\noindent
For every adjacent pair of \texttt{even} and \texttt{odd} branches,
\texttt{RoundFunctions} loads the two corresponding 128-bit pieces of the
state and overwrites \texttt{odd} with \texttt{F(even, key) XOR odd}.

{\small
\begin{lstlisting}[numbers=none]
  const V* RoundFunctions(const V* keys) {
    for (int branch = 0; branch < 16; branch += 2) {
      const V even = Load(state, branch);
      const V odd = Load(state, branch + 1);
      const V new_odd = F(even, *keys++) ^ odd;
      Store(new_odd, state, branch + 1);
    }
    return keys;
  }
\end{lstlisting}}

\noindent
Note that the XOR can be computed for free because the last step of AES simply
XORs with its round constant. We change the second AES round constant in
\texttt{F} from \texttt{zero} (which has no effect) to \texttt{odd}. The
\texttt{key} passed to each call to \texttt{F} comes from an array of eight
AES keys.

The second layer of the Feistel network (denoted \texttt{BlockShuffle})
rearranges the 128-bit branches into the prescribed order
\cite{improvedFeistel}[p.~21, no.~10]. We permute the state such that the
previous branch 7 comes first, followed by 2, 13, 4 and so on
(see \texttt{shuffle} below):

{\small
\begin{lstlisting}[numbers=none]
  void BlockShuffle() {
    uint64_t source[32];
    memcpy(source, state, sizeof(source));

    constexpr int shuffle[16] = {
      7, 2, 13, 4, 11, 8, 3, 6, 15, 0, 9, 10, 1, 14, 5, 12};
    for (int branch = 0; branch < 16; ++branch) {
      const V v = Load(source, shuffle[branch]);
      Store(v, state, branch);
    }
  }
\end{lstlisting}}

\noindent
Together, these two layers constitute one round of a generalized Feistel
network. The final permutation \texttt{Permute} consists of 17 rounds. Each
invocation of the \texttt{RoundFunctions} layer requires eight AES round
constants, for a total of 2176 bytes.

{\small
\begin{lstlisting}[numbers=none]
  void Permute() {
    // Round keys for one AES per Feistel round and branch.
    const V* keys = Keys();

    for (int round = 0; round < 17; ++round) {
      keys = RoundFunctions(keys);
      BlockShuffle();
    }
  }
\end{lstlisting}}

\noindent
The \texttt{keys} can be a fixed array of nothing-up-my-sleeve numbers shared by
all generators. However, our indistinguishability result
(Section~\ref{sec:indist}) assumes a keyed/secret permutation,
otherwise attackers could distinguish the permutation from random by
querying it and verifying the expected result. Applications running on
secure servers may reasonably expect that attackers do not have access to
the key. For additional safety, applications could instead generate
the keys at startup using a stream cipher such as ChaCha20 keyed with 256 bits
obtained from a trusted source such as the operating system.
Note that the generator remains backtracking-resistant
(Section~\ref{sec:backtracking}) even if the keys are leaked.

\subsection{Generator}

Now that we have defined the permutation, Reverie's \texttt{Generate} produces
random outer bits by invoking  \texttt{Permute} on the state and XORing the
inner bits with the value they had before the permutation, which cannot be
reversed by an attacker with knowledge of the current state \cite{reverie}:
{\small
\begin{lstlisting}[numbers=none]
  void Generate() {
    const uint64_t prev_inner[2] = { state[0], state[1] };

    Permute();

    // Ensure backtracking resistance.
    state[0] ^= prev_inner[0];
    state[1] ^= prev_inner[1];
  }
\end{lstlisting}}

\noindent
As a result, the last 1920 bits of \texttt{state} are uniform random and
available for use. In practice, the generator is packaged as a C++ `random
engine' that returns 32 or 64-bit bundles of random bits and calls
\texttt{Generate} again once all remaining bits have been consumed.

\section{Rationale}
\label{sec:rationale}

Here we briefly justify design decisions.

\begin{description}
\item[The AES block cipher] is well-understood and often hardware-accelerated.
Intel's AESNI instructions \cite{intelAESNI} are five to ten times faster than
optimized software implementations \cite{kasperAES, softwareAES2}. This
implies a software-only Randen would be unacceptably slow (and likely
vulnerable to side-channel attacks). On CPUs without hardware AES, it may be
faster to replace AES with SIMD-friendly permutations such as ChaCha
\cite{ebacsStream}. However, most modern CPUs have AES hardware, including
POWER (\texttt{VCIPHER} \cite{powerAES}) and ARMv8 (\texttt{AESE}
\cite{armManual}).

\item[Two AES rounds] are necessary for full-bit diffusion \cite{simpira} and
more efficient than a single round in terms of the ratio of active
s-boxes\footnote{in a standard type-2 Feistel with four branches.}
\cite{doubleSP}.

\item[Dense and independent AES round keys] ensure that an AES round breaks
the symmetry of plaintext with all-equal columns. We use unique keys to rule
out attacks similar to those on Haraka~v1 \cite{harakaV2} and Simpira~v1
\cite{simpira}. This requires a total of 2176 bytes, which is somewhat
excessive, but the keys are typically hardcoded (but not necessarily public)
nothing-up-my-sleeve numbers and there is little cost to loading unique keys
because they easily fit in the L1 cache.

\item[Type-2 generalized Feistel networks] are often used to construct large
permutations from smaller blocks. These constructions are `sound' in the sense
that they are strong pseudorandom permutations after sufficient rounds of a
pseudorandom function \cite{improvedFeistel}. In contrast to the $b > 8$
variants of Simpira~v2 \cite{simpira}, they enable good performance without
relying on multiple independent inputs to keep the CPU pipeline filled.

\item[An improved block shuffle] for the generalized Feistel network reaches
full sub-block diffusion (i.e.\ each block depends on every other input block)
much sooner than traditional cyclic shifts \cite{improvedFeistel}. It also
reduces vulnerability to sliced-biclique \cite{newDiffusion} and integral
attacks \cite{genFeistelEval}[p.~226].

\item[16-branch generalized Feistel networks] are the largest for which the
diffusion properties are known \cite{newDiffusion}. Larger branch counts have
two related benefits without requiring multiple independent inputs like
Simpira~\cite{simpira}. First, they enable parallel evaluation of the round
functions, which hides the long latency of \texttt{AESENC} \cite{intelOpt}.
Second, they can benefit from increased hardware parallelism such as recently
announced quadruple-AES hardware \cite{intelExtensions}[p.~2-14].

\item[A 2048-bit permutation] is a natural result of 16-branch Feistel with
128-bit AES blocks. Larger states cannot be accommodated within the 16 SSE4
registers.

\item[17 Feistel rounds] improve the diffusion relative to the minimum of
16~rounds required for Feistel block diffusion (propagating input differences
to each branch of the state) \cite{improvedFeistel}.

\item[Reverie] is an efficient construction for backtracking-resistant
generators. It avoids the heavy rekeying cost of CTR-DRBG and exposes
fewer prior outputs than an only periodically re-keyed stream cipher.

\item[Reseeding] the state from external entropy sources periodically is
beyond the scope of this paper because our applications typically require
reproducible sequences of random numbers.

\end{description}

\section{Implementation details}
\label{sec:impl}

We implement the algorithm in C++ using SIMD intrinsics that are available on
current Intel, AMD and POWER CPUs. The final optimized code is quite short
(only about 150 lines) and very similar to the straightforward listings above!
If \texttt{state} is a restrict-qualified pointer, Clang understands that
\texttt{BlockShuffle} simply renames memory locations. We have released this
code \cite{randenGithub} under an open-source license so our results can be
reproduced.

In the rest of this section, we study how well the algorithm maps to the
Haswell and Skylake microarchitectures. Despite the high-level implementation,
the measured \texttt{Permute} throughput is within 5\% of the lower bound (one
\texttt{AESENC} per cycle). Intel's IACA simulator \cite{intelIACA} reports
the code is bottlenecked by the `frontend' in addition to the expected port 5
(\texttt{AESENC}), but still claims its throughput should exactly match the
lower bound. Note that IACA does not model memory accesses, and the limited
set of 16 SSE4 registers necessitates many spills to memory, so the 5\%
difference is probably due to loads. However, we also investigate the alleged
frontend limitation using performance counters captured via the Linux perf
utility.

Is decode throughput the bottleneck? This can be a problem because the 16 byte
fetch window (unchanged since Pentium Pro) is too small for large SIMD
instructions (\emph{nine} bytes for \texttt{AESENC} with a 32-bit offset). Two
such instructions do not fit in a fetch window, so only one can decode per
cycle. However, Sandy Bridge and later Intel CPUs include a decoded
instruction cache (DSB), which is very helpful because it avoids the 16-byte
limitation. Indeed we find 99.9\% of \textmu ops are delivered from the DSB.
However, the effective DSB capacity is lower than the documented maximum of
1536 \textmu ops. Fully unrolling the Feistel rounds generates about
750~\textmu ops and causes a 10x increase in DSB misses. Unrolling by a factor
of two generates good code.

Is microcode a factor? In 2012 there was speculation that \texttt{AESENC}
uses the microcode sequencer (MSROM) \cite{aesMSROM}. We can confirm this is
not the case (on Haswell) because \verb|IDQ.MS_UOPS (79_30)| is zero.
Given the low values of \verb|IDQ_UOPS_NOT_DELIVERED.CORE (9C_01)|, we can
conclude the bottleneck does not involve the decoders.

What about other stalls? \verb|LD_BLOCKS_PARTIAL.ADDRESS_ALIAS (07_01)|
detects 4K aliasing between compiler-generated spills to the stack and loads
of round keys. This is difficult to reliably avoid, but only affects 1\% of
all instructions. \verb|RESOURCE_STALLS (A2_FF)| affect 18\% of all
instructions; 90\% of these are waiting for the reservation station. We
speculate that this is due to a lack of physical registers and/or waiting for
loads. Either way, the problem should disappear on Skylake. With its 32 vector
registers, we can devote 8 to the AES inputs and outputs (updated in-place)
and 8+8 to hold the XOR inputs for the next two rounds, thus entirely avoiding
spills. In summary, it appears difficult to further optimize the
implementation. We emphasize that the compiler and out-of-order CPU extract
good performance (within 5\% of the lower bound) from our minimally annotated
high-level language implementation.

\section{Smoke test}

Every random generator should avoid `recognizable patterns', which can cause
systematic errors in applications such as simulations \cite{rngCorrelation}.
In the next section, we argue Randen is computationally indistinguishable from
random, which implies the non-existence of any patterns. However,
general-purpose generators are unable to furnish such arguments, so they
instead apply statistical tests to detect obvious flaws. Several batteries of
tests are well-known and often used for verifying empirical randomness. For
completeness, we also apply them to Randen. We begin with BigCrush from
TestU01 version 1.2.3 \cite{testu01}. Its interface requires a small wrapper
around the raw generator \cite{testu01Guide}:

{\small
\begin{lstlisting}[numbers=none]
  randen::Randen<uint32_t> engine;
  uint32_t Rand32() { return engine(); }

  int main(int, char*[]) {
    unif01_Gen* gen = unif01_CreateExternGenBits("R", Rand32);
    bbattery_BigCrush(gen);
    unif01_DeleteExternGenBits(gen);
    return 0;
  }
\end{lstlisting}}

\noindent All 160 tests pass for PCG \cite{pcg} and Randen with original and
inverted bits. By contrast, BigCrush reports two failures when testing MT19937
and one near-failure for AES-CTR (p-value of 0.000092, but it did not recur
in subsequent test(s)) \cite{testu01}.

We also test Randen with the current version 0.93 of PractRand
\cite{practrand}. To avoid file or pipe overhead, we integrate Randen into the
\texttt{DummyRNG} class by having its \texttt{raw32} function return Randen's
output. The test battery is invoked with default settings via
\verb|./RNG_test dummy -multithreaded|. Running all tests up to the upper limit
of 32~terabytes reports two `unusual' p-values (0.9921 and 0.0013). Note that
\texttt{pcg64} also leads to an unusual p-value (0.0016) in a much smaller
test, and failures have more extreme p-values, e.g.\ $10^{-351}$ for Mersenne
Twister \cite{practrandXORO}. We conclude that Randen passes state of the art
tests of empirical randomness about as well as \texttt{pcg64} and better than
Mersenne Twister.

\section{Security}
\label{sec:security}

Some developers are unaware that randomized applications can be vulnerable to
adversaries and we have observed reluctance to sacrifice speed for security. It
is expensive to audit tens of thousands of random generator usages to
determine the appropriate security/speed tradeoff. We therefore propose to
provide a higher baseline level of security than existing general-purpose
generators. To gain user acceptance, we ensure our generator remains within the
performance envelope of Mersenne Twister. What security guarantees can we
provide? In this paper, a `strong' generator is characterized by two properties:
computational indistinguishability from random, and backtracking resistance. In
the following, we show that these hold for Randen.

\subsection{Indistinguishability}
\label{sec:indist}

`Indistinguishable from random' is a very strong property often used in
cryptography. We emphasize that security-critical applications should continue
to use trusted cryptographically secure generators. However, other
applications also benefit from a strong generator. Indistinguishability
implies the output is unpredictable, which prevents adversaries from
triggering worst case execution time in randomized algorithms such as
Quicksort (quadratic rather than linearithmic time), or influencing the
samples drawn by randomized online sampling algorithms.

We now apply a standard computational indistinguishability argument. Suppose a
deterministic adversary is given query access to either a real or ideal (i.e.\
uniform random) generator and returns 0 or 1 to indicate which generator it is
interacting with. We assume an adversary cannot issue more than $2^{64}$
permutation queries.  Then, a real generator is computationally
indistinguishable from random if the distinguishing advantage (absolute
difference in probability of any such adversary returning 1 when given the
ideal vs.\ real generator) is negligible.

\begin{lemma}
\label{lem:randenIndist}
In the ideal permutation model, if the Randen permutation is replaced with an
ideal permutation, Randen is indistinguishable from random by adversaries
limited to $2^{64}$ permutation queries.
\end{lemma}
\begin{proof}
Randen is an instantiation of Reverie, which guarantees that the
best possible attack must guess its inner bits \cite{reverie}[p.~12]. That
requires an average of $2^{127}$ evaluations of the Randen permutation, which
is beyond the capabilities of our assumed adversary.
\end{proof}

There are two practical difficulties with the ideal permutation model. First,
attackers can trivially distinguish a Randen permutation with known key by
simply querying it. In this section, we need to assume the permutation is
keyed. Second, a truly random permutation is impractical because its
representation requires $\log_2{2^{128}!} \approx 10^{40}$ bits. We could
instead argue that the generalized Feistel structure of the Randen permutation
ensures it would be indistinguishable from random if its round functions were
pseudorandom \cite{improvedFeistel}. However, our round function consists of
two rounds of AES, and up to three are efficiently distinguishable from random
\cite{designAES}. We could construct a round function that is believed to be
indistinguishable from a random function by XORing two permutations
\cite{sumPRP} that are widely recognized to be secure, such as 10~rounds of
AES. Unfortunately this would be about ten times slower. Instead, we will
study known attacks on the actual Randen permutation.

\begin{table}[h]
\caption{Lower bound on active functions after a given number of rounds of a
16-branch type-2 Feistel network with improved block shuffle. Derived via
exhaustive search in Appendix~\ref{sec:search}.}
\label{tab:active}
\centering
\begin{tabular}{r|l|r|l}
Rounds & Active Functions & Rounds & Active Functions\\
\midrule
1  &  0 & 13 & 27\\
2  &  1 & 14 & 30\\
3  &  2 & 15 & 32\\
4  &  3 & 16 & 35\\
5  &  4 & 17 & 36\\
6  &  6 & 18 & 39\\
7  &  8 & 19 & 41\\
8  & 11 & 20 & 44\\
9  & 14 & 21 & 45\\
10 & 18 & 22 & 48\\
11 & 22 & 23 & 50\\
12 & 24 & 24 & 53\\
\end{tabular}
\end{table}

The security of Substitution-Permutation (SP) networks such as AES is often
established by showing sufficiently many s-boxes are active to resist
differential and linear attacks \cite{simpira}. Such results are also available
for generalized Feistel networks, but they are specific to the number of
branches and type of round function. We use 16 branches and SPSP-type
functions (two rounds of AES). Existing results are available for either
situation, but not both. 6 rounds of SPSP functions in a 4-branch type-2
network guarantee 6 differentially active \emph{functions} \cite{doubleSP}.
17 rounds of SP functions in a 16-branch network with improved diffusion
guarantee 78 active s-boxes \cite{feistelSboxes}[p.~226]. Later in this section,
we provide new results for 16-branch networks with SPSP functions.

Note that 16-branch Feistel networks have a maximum impossible differential
characteristic of 14 rounds \cite{improvedFeistel}, and the sliced biclique
technique only attacks 15 rounds \cite{newDiffusion}. A recent attempt to find
integral distinguishers reports `difficulty' for such large branch counts
\cite{genFeistelEval}[p.~219]. We compute new lower bounds for active
functions in 16-branch type-2 Feistel networks via exhaustive search. Details
of the algorithm are deferred to Appendix~\ref{sec:search}. The resulting
lower bounds are given in Table~\ref{tab:active}. Note that we are able to
compute bounds for up to 24 rounds, whereas prior results for 16-branch
Feistel networks only extend to 18 rounds \cite{feistelSboxes}. A meet in the
middle attack \cite{twineMITM} splits a permutation into three parts. Hence,
we consider the number of active functions after six rounds.

\begin{theorem}
\label{the:diff}
The probability of differential characteristics and correlation of
linear characteristics of six rounds of the Randen permutation are at most
$2^{-180}$ and $2^{-90}$.
\end{theorem}
\begin{proof}
Per Table~\ref{tab:active}, at least six functions are active after six
rounds. Each active SPSP function provides at least $\mathcal{B}(M)$ active
s-boxes \cite{doubleSP}. $\mathcal{B}(M)$ is the branch number of the SP
permutation layer, which is 5 for AES. Thus, at least 30 s-boxes are active.
Each active AES s-box contributes a maximum differential probability $2^{-6}$
and correlation amplitude $2^{-3}$ \cite{designAES}. Thus, the overall
differential probability and linear correlation are $2^{-6 \cdot 30}$ and
$2^{-3 \cdot 30}$.
\end{proof}
\noindent
Note that Simpira's security arguments only require 25 active s-boxes
\cite{simpira}. Also, Table~\ref{tab:active} indicates there are
$36 \cdot \mathcal{B}(M) = 180$ active s-boxes after our 17 rounds with
SPSP functions. By contrast, the prior bound for SP functions only guarantees
78~active s-boxes after 17~rounds \cite{feistelSboxes}[p.~226].

\begin{claim}
A keyed Randen permutation cannot be distinguished from random with complexity
less than $2^{64}$.
\end{claim}
This bound is a conservative estimate based on our initial analysis. Per
Theorem~\ref{the:diff}, differential/linear attack complexity is $2^{180}$ and
$2^{90}$. Symmetry attacks on AES are also unlikely to succeed
because our round keys lack structure. Note that Randen involves 17 AES
subrounds per 16 permuted bytes, versus only 10 for the AES-128 cipher. Any
distinguishers would seem to imply new (or unknown to us) attacks on
generalized Feistel with AES-like rounds.

For comparison, a recent successful attack on the full SHA-1 involved $2^{63}$
work at an estimated cost of 110,000 USD \cite{sha1Attack}. We assume this is
a sufficient deterrent to predicting outputs.

\begin{lemma}
If a computationally bounded adversary cannot distinguish the Randen
permutation from random, then they cannot predict the Randen output with less
than $2^{64}$ work based only on prior outputs.
\end{lemma}
\begin{proof}
In this setting, adversaries do not know the AES round keys. The only way
adversaries can access the permutation is by requesting random output. We can
meet the requirements of Lemma~\ref{lem:randenIndist} by instantiating Randen
with an oracle implementing a randomly keyed Randen permutation. From the
perspective of the adversary, this behaves in the same way as Randen
instantiated with a real permutation. Then, the Randen output is
indistinguishable from random, which implies unpredictability by contradiction
(predicting a future output would also allow an adversary to distinguish the
generator from random).
\end{proof}

Note that ``based only on prior outputs'' excludes cases where attackers gain
access to the inner state. By contrast, NIST~800~90a requires prediction
resistance even after the state is compromised \cite{nistRequirements}. This
would entail periodic reseeding from external entropy sources, which we must
avoid to ensure repeatability. Instead, we note that side-channels such as
core dumps and paging \cite{stateSecrecy} are less relevant in a server
environment and can be mitigated with the help of the operating system (using
\texttt{madvise} and \texttt{mlock}). Then, attackers can only guess the inner
state at a cost of $2^{127}$, which is beyond their assumed capability.

\subsection{Backtracking resistance}
\label{sec:backtracking}

The second property is backtracking resistance: adversaries have a negligible
advantage at distinguishing prior outputs from random even if they gain access
to the state \cite{onBacktracking}. This is important for portable devices and
long-running applications without access to external entropy because it
ensures adversaries cannot reconstruct prior outputs. If a generator is
robust, it also provides backtracking resistance (also known as forward
security) \cite{dodisRequirements}. Reverie is robust in the ideal permutation
model \cite{reverie} and the previous section argues that instantiating
Reverie with a random Randen permutation retains its security guarantees.
However, robustness requires reseeding the generator from external entropy,
which is not always possible in our applications. We instead show that
backtracking resistance follows from the security of Reverie's next function
\cite{reverie}[p.~12], i.e.\ Randen's \texttt{Generate}. Assume an adversary
has gained access to the current \texttt{state} and AES keys. This allows them
to invert the Randen permutation. Note that Reverie's security model assumes a
public permutation that attackers can already invert. We will illustrate the
backtracking resistance in a scenario with two calls to \texttt{Generate}.
Additional calls do not affect the argument.

For the following, let us define new notation: the state after the first
($k=1$) and second ($k=2$) call to \texttt{Generate} can be partitioned into
inner/outer parts $i_k$ and $o_k$. Let $i_0$ and $o_0$ denote the uniform
random initial state. Per Lemma~1 of Reverie \cite{reverie}[p.~12], the return
values of \texttt{Generate} are indistinguishable from random. However, the
attacker knows all random outputs (i.e.\ outer states $o_0,o_1$) and learns
the current state $i_2,o_2$. Does this allow them to recover the remaining
prior inner states $i_0,i_1$? Recall that the final \texttt{Generate} returns
\texttt{Permute}($i_1,o_1$) XOR ($i_1$, zero). All terms except $i_1$ are
known. However, the attacker cannot query the permutation in either direction
without guessing the $i_1$ value at a cost of $2^{127}$; this is best possible
attack on Reverie \cite{reverie}. Hence, knowledge of $o_0,o_1,o_2,i_2$ is
insufficient, and adversaries cannot expect to distinguish prior outputs from
random with less than $2^{127}$ forward or backward queries to the permutation
(e.g.\ by guessing the inner bits) \cite{reverie}[p.~12]. Therefore, Randen is
backtracking-resistant.

\section{Performance} %
\label{sec:speed} %

\subsection{Contenders}

We emphasize that our comparison involves three groups of generators,
in increasing order of security.

\subsubsection*{Insecure generators}
To establish a performance baseline, we include the commonly used but insecure
Mersenne Twister (`MT') as implemented by the C++11 standard library. Note that
faster variants of MT exist \cite{fastMT, avxMT}. However, we advocate using
more secure generators in most applications with the exception of Monte Carlo
simulations.

\subsubsection*{Medium-strength}
Several recent generators are at least nontrivial to distinguish from random,
although indistinguishability and backtracking-resistance have not been
formally shown. We include `Philox' \cite{salmonRNG} and \verb|pcg64_c32|
\cite{pcg} (`PCG'), both of which make no concrete security claims. We also
place ISAAC into this category -- although no bias has been shown, there are
doubts about its security and similarity to RC4.

\subsubsection*{Strong}
The third group consists of generators with security claims (see
Section~\ref{sec:security}). In addition to Randen, we include `ChaCha20'
(provided by Linux 4.9 \verb|/dev/urandom| \cite{linuxChaCha}) and `CTR-DRBG'
from NIST~SP~800-90A (provided by Windows 7 \texttt{BCryptGenRandom}). These
have higher overhead, possibly due to calling into kernel mode. We reduce this
somewhat by using a 256-byte buffer, the same size as Randen. To fully exclude
the OS overhead, we also include a user-mode SSE2 implementation of ChaCha8 by
Orson Peters that uses a single 64-byte block. Note that Bernstein recommends
ChaCha20 instead due to its higher security margin \cite{djbTwitter}.

\subsection{Infrastructure}

All generators except `CTR-DRBG' are implemented in C++ and compiled using
Clang \texttt{r331746} with {\small\texttt{-O3 -std=gnu++11}}. The `x86'
benchmark is pinned to a single core of a lightly loaded dual-socket Xeon
E5-2690~v3 clocked at 2.6~GHz running Linux 4.9 with Turbo Boost and throttling
disabled. We also report performance on a POWER~8e clocked at 3.6~GHz (`PPC').
The `CTR-DRBG' measurements are obtained on an Intel i7~4790K CPU clocked at
4.0~GHz running Windows 7 x64 and using the Microsoft Visual Studio 2017
compiler. To increase the precision and accuracy of generator speed
measurements, we use an
\ifIsSubmission
improved benchmarking infrastructure.
\else
improved version of the `nanobenchmark' infrastructure
\cite{nanobenchmark} developed for HighwayHash.
\fi
It prevents elision of the generator by passing its output as
an input to an empty inline assembly block marked as modifying memory. To
reduce variability between runs, it records high-resolution timestamps (in
units of CPU cycles) from the invariant TSC, uses fences to ensure the
measured code is not reordered by the compiler nor CPU, subtracts the overhead
of the TSC reads and uses the median (for small sample counts) or mode as a
robust estimator of the central
tendency. As a result, variability between measurements (defined as median
absolute deviation from the median) is about 0.2\%. To improve
comparability between
benchmarks of different sizes, we divide the elapsed times by the number of
random bytes generated to yield cycles per byte. Note that the PPC elapsed
times are relative to its 512~MHz timebase, so we multiply measurements by
$7 \approx 3600 / 512$ to obtain CPU cycles.

\subsection{Benchmarks}

We go beyond conventional microbenchmarks by including three simple real-world
applications of random numbers: a Fisher-Yates shuffle
\cite{fisherYatesShuffle}, reservoir sampling \cite{reservoirSampling}, and a
Monte Carlo estimator for the value of $\pi$. Together, these exercise all
consumers of random bits in the C++ standard library.

We emphasize that our measurements encompass the entire application, viz.: the
algorithm consuming random numbers (e.g. shuffling) plus buffer-empty checks
required by the C++ random generator interface plus the generator itself.
Thus, the reported throughput will naturally be lower than best-case
microbenchmarks of a stream cipher or merely generating large quantities of
random bits.

\subsubsection{Microbenchmark}

C++11 only requires amortized constant-time complexity for its uniform random
generators. This allows them to return numbers from a large buffer
which is periodically refilled. To measure the amortized cost, we must ensure
the elapsed time measurements include sufficient refills. Although the buffer
sizes are known, C++11 does not provide a guaranteed means of flushing or
querying the buffer. We therefore generate 800~KB of random bits such that the
cost of `wasting' part of the final buffer is negligible.

\begin{table}[H]
\caption{Cycles per byte for a small loop, plus variability (MAD is
the median absolute deviation) and speedup factor of Randen vs.\ other
generators.}
\label{tab:loop}
\centering
\begin{tabular}{r|r|r||r|r}
Engine & x86 (MAD) & Speedup & PPC (MAD) & Speedup\\
\midrule
Randen   &  1.54 ($\pm$ 0.002) &  --  &   2.94 ($\pm$ 0.007) &  --  \\
PCG      &  0.78 ($\pm$ 0.003) &  0.5 &   1.68 ($\pm$ 0.007) &  0.6 \\
MT       &  1.79 ($\pm$ 0.001) &  1.2 &   3.99 ($\pm$ 0.014) &  1.4 \\
ChaCha8  &  3.02 ($\pm$ 0.003) &  2.0 &                    &      \\
ISAAC    &  4.08 ($\pm$ 0.006) &  2.6 &   7.91 ($\pm$ 0.014) &  2.7 \\
Philox   &  4.70 ($\pm$ 0.003) &  3.1 &   9.94 ($\pm$ 0.014) &  3.4 \\
ChaCha20 & 15.27 ($\pm$ 0.018) &  9.9 & 197.96 ($\pm$ 0.315) & 67.3 \\
CTR-DRBG & 16.80 ($\pm$ 0.009) & 11.2 &                    &      \\
\end{tabular}
\end{table}

\noindent The x86 microbenchmark (Table~\ref{tab:loop}) seems to indicate PCG
is twice as fast as Randen, which is in turn 1.2 times as fast as MT. The
trend is similar on PPC. Despite their high precision (median absolute
deviation below 0.2\%), these microbenchmark results are quite irrelevant in
practice --- which actual application repeatedly calls a random generator and
ignores the results? Any that do should use the more efficient
\texttt{discard} function instead. As we will see, these results are not
representative of real-world performance. There are at least three reasons why
microbenchmarks may mischaracterize actual performance. First, their small
working set leads to unrealistically high cache and TLB hit rates. Second,
tight loops benefit from special CPU decoding hardware
\cite{agnerMicroarch}[p.~123]. Third, simple microbenchmarks may use fewer
CPU resources (e.g.\ registers and load-store buffers) than real-world
applications.

\subsubsection{Shuffle}

For a more realistic use case, we measure a Fisher-Yates shuffle that swaps
elements at a randomly chosen position. Although the C++ standard library
provides an implementation (\verb|std::shuffle|), its mapping of random bits
to uniform integers is quite slow. Instead of costly divisions, we use a
multiplication followed by bit-shift \cite{lemireMod}. The resulting shuffle
is about three times as fast as \verb|std::shuffle|. The array is 400~KB
large, which exceeds the 256~KiB L2 cache on x86 but fits into the 512~KiB PPC
cache.

We observe different performance characteristics (Table~\ref{tab:shuffle})
than in the microbenchmark. As shown by the `Randen factor' columns, Randen is
1.2 times as fast as PCG and slightly faster than MT on x86. By contrast, MT
is the fastest on PPC. In all benchmarks, Randen is roughly twice as fast as
ISAAC, which is still faster than Philox.
Replacing ChaCha20/CTR-DRBG with Randen leads to an
overall shuffle speedup of 7 to 8, and 36 on PPC. We see nearly identical
results on x86 for 100~KB and 25~KB inputs, which fit into L2 and L1,
respectively. This implies that caching and prefetching are effective. Indeed
VTune reports that only Philox and ChaCha20 have high levels of load/store
stalls: 45\% and 85\%, versus less than 30\% for the other generators.

\begin{table}[H]
\caption{Cycles per byte for engines called from Fisher-Yates shuffle, plus
variability (MAD is the median absolute deviation) and speedup factor of
Randen vs.\ other generators.}
\label{tab:shuffle}
\centering
\begin{tabular}{r|r|r||r|r}
Engine & x86 (MAD) & Speedup & PPC (MAD) & Speedup\\
\midrule
Randen   &  2.19 ($\pm$ 0.004) & --  &   5.46 ($\pm$ 0.014) &  --  \\
PCG      &  2.65 ($\pm$ 0.005) & 1.2 &   6.65 ($\pm$ 0.014) &  1.2 \\
MT       &  2.19 ($\pm$ 0.004) & 1.0 &   4.48 ($\pm$ 0.021) &  0.8 \\
ChaCha8  &  3.63 ($\pm$ 0.006) & 1.7 &                    &      \\
ISAAC    &  4.15 ($\pm$ 0.007) & 1.9 &   8.19 ($\pm$ 0.021) &  1.5 \\
Philox   &  4.87 ($\pm$ 0.008) & 2.2 &  10.57 ($\pm$ 0.021) &  1.9 \\
ChaCha20 & 15.87 ($\pm$ 0.027) & 7.2 & 198.24 ($\pm$ 0.917) & 36.3 \\
CTR-DRBG & 20.45 ($\pm$ 0.017) & 8.2 &                    &      \\
\end{tabular}
\end{table}

\subsubsection{Sample}

Our third benchmark measures reservoir sampling, a randomized
online algorithm for retaining an 80~KB subset of a 400~KB data stream. It
probabilistically overwrites prior samples at random position. As with
shuffling, using a division-free mapping of random bits to integers is much
faster than \verb|std::uniform_int_distribution|. We see similar performance
(Table~\ref{tab:sample}), except that Randen is now 1.2 times as fast as PCG on
x86, and 1.4 on PPC. On both platforms, Randen outperforms MT.
Also as before, speeds are comparable when reducing the input sizes to
one quarter.

\begin{table}[H]
\caption{Cycles per byte for engines called from reservoir sampling, plus
variability (MAD is the median absolute deviation) and speedup factor of
Randen vs.\ other generators.}
\label{tab:sample}
\centering
\begin{tabular}{r|r|r||r|r}
Engine & x86 (MAD) & Speedup & PPC (MAD) & Speedup\\
\midrule
Randen   &  2.60 ($\pm$ 0.008) & --  &   4.97 ($\pm$ 0.007) &  --  \\
PCG      &  3.03 ($\pm$ 0.009) & 1.2 &   6.72 ($\pm$ 0.021) &  1.4 \\
MT       &  2.82 ($\pm$ 0.009) & 1.1 &   5.32 ($\pm$ 0.014) &  1.1 \\
ChaCha8  &  3.75 ($\pm$ 0.008) & 1.4 &                    &      \\
ISAAC    &  4.46 ($\pm$ 0.014) & 1.7 &   8.12 ($\pm$ 0.014) &  1.6 \\
Philox   &  4.95 ($\pm$ 0.009) & 1.9 &   9.87 ($\pm$ 0.007) &  2.0 \\
ChaCha20 & 13.46 ($\pm$ 0.017) & 5.2 & 159.67 ($\pm$ 0.168) & 32.1 \\
CTR-DRBG & 16.41 ($\pm$ 0.015) & 6.4 &                    &      \\
\end{tabular}
\end{table}

\subsubsection{Monte Carlo}

The fourth benchmark is Monte Carlo estimation of the value of $\pi$ via the
ratio of points that fall within a unit circle versus the unit square. This is
similar to the microbenchmark in that it calls the generator 200,000 times in
a fairly tight loop. Note that \verb|std::uniform_real_distribution| is slow
and not actually uniform \cite{uniformBug}, so we again implement a
replacement. It constructs an IEEE-754 mantissa using the lower 53 bits of a
generated \verb|uint64_t| and chooses an exponent based on the base-2
logarithm of its upper bit. The results in Table~\ref{tab:montecarlo} show
that PCG is 1.2 times as fast as Randen on x86 but slower on PPC. Randen
outperforms MT on both platforms.

\begin{table}[H]
\caption{Cycles per byte for engines called from a Monte Carlo simulation,
plus variability (MAD is the median absolute deviation) and speedup factor
of Randen vs.\ other generators.}
\label{tab:montecarlo}
\centering
\begin{tabular}{r|r|r||r|r}
Engine & x86 (MAD) & Speedup & PPC (MAD) & Speedup\\
\midrule
Randen   &  2.14 ($\pm$ 0.002) & --  &   3.43 ($\pm$ 0.007) &  --  \\
PCG      &  1.69 ($\pm$ 0.031) & 0.8 &   3.85 ($\pm$ 0.007) &  1.1 \\
MT       &  2.55 ($\pm$ 0.015) & 1.2 &   4.90 ($\pm$ 0.007) &  1.4 \\
ChaCha8  &  4.58 ($\pm$ 0.002) & 2.1 &                    &      \\
ISAAC    &  4.35 ($\pm$ 0.003) & 2.0 &   8.54 ($\pm$ 0.056) &  2.5 \\
Philox   &  4.97 ($\pm$ 0.002) & 2.3 &  11.62 ($\pm$ 0.014) &  3.4 \\
ChaCha20 & 16.65 ($\pm$ 0.006) & 7.8 & 194.53 ($\pm$ 8.428) & 56.7 \\
CTR-DRBG & 17.37 ($\pm$ 0.031) & 9.3 &                    &      \\
\end{tabular}
\end{table}

\subsection{Discussion}

To summarize the four benchmarks, we compute the geometric means of the
`Randen factors' from the above tables, i.e.\ the cost (cycles per byte) of
other generators divided by that of Randen (Table~\ref{tab:geomean}).
Due to the large differences in the (lack of) security guarantees of the
various generators, we discuss them separately.

\begin{table}[H]
\caption{Geometric means of Randen speedup factors across the benchmarks. A
value of 1.1 indicates the benchmarks run 1.1 times as fast after replacing MT
with Randen.}
\label{tab:geomean}
\centering
\begin{tabular}{r|r|r}
Engine & x86 & PPC\\
\midrule
PCG      & 0.9 &  1.0\\
MT       & 1.1 &  1.1\\
ChaCha8  & 1.8 &     \\
ISAAC    & 2.0 &  2.0\\
Philox   & 2.3 &  2.6\\
ChaCha20 & 7.3 & 45.9\\
CTR-DRBG & 8.5 &     \\
\end{tabular}
\end{table}

\subsubsection*{Insecure generators}
One of our main results is that the Randen generator does not increase
CPU cost relative to the commonly used but insecure Mersenne Twister
generator. The geometric mean of speed ratios indicates Randen is
slightly faster on both x86 and PPC.

\subsubsection*{Medium-strength}
Randen is about twice as fast as ISAAC and Philox in all benchmarks. Our choice
of geometric mean indicates PCG is the fastest on x86, and tied for first on
PPC.
However, this is mainly due to its result in the (unrealistic) microbenchmark.
PCG is a good choice for Monte Carlo applications but Randen is 1.2 to 1.4
times as fast for shuffling and sampling. Note that ISAAC, Philox and PCG lack
concrete security claims and have not been shown to be indistinguishable from
random nor backtracking-resistant. To the best of our knowledge, Randen is the
fastest software generator with these properties.

\subsubsection*{Strong}
Is it feasible to use cryptographically secure generators as the default
even in non-cryptographic applications? This depends on the scale of usage.
We profiled production code running company-wide and found that traditional
non-cryptographic random
generators account for tens of thousands of CPU cores. From this and the
above benchmarks, we conclude it would be too expensive to use an OS-provided
ChaCha20 (\verb|/dev/urandom|) or CTR-DRBG (\texttt{BCryptGenRandom}) as
general-purpose generators. By contrast, Randen is 5 to 10 times as fast in
real-world benchmarks. Switching from Mersenne Twister to Randen would
actually reduce cost (according to the geometric mean of our benchmarks), while
greatly increasing the baseline security of non-cryptographic randomized
applications.

\section{Conclusion}

Recent random generators have desirable characteristics: SIMD-accelerated
Mersenne Twister (MT) is efficient \cite{avxMT}. PCG has good
statistical properties \cite{pcg}. AES-CTR is
unpredictable by attackers. AES-CTR-DRBG ensures backtracking resistance
\cite{nistRequirements}. Thanks to recent hardware acceleration of AES, a
single generator can now achieve all these goals!

This work proposes \textbf{Randen}, an instantiation of Reverie \cite{reverie}
with a permutation based on a generalized Feistel structure. We show that it
is `strong', i.e.\ computationally indistinguishable from random and
backtracking resistant. This high level of security is useful even for
general-purpose applications such as shuffling and sampling because it greatly
increases the attacker cost of triggering worst-case behavior in randomized
algorithms. Note that Randen is not intended for cryptographic applications such
as key generation, but the permutation may also be useful for wide-block ciphers
and hashing functions.
Despite its statistical quality and resistance to attacks, Randen is actually
faster than the commonly used MT generator, ChaCha8, ISAAC, Philox and a
variant of PCG in some real-world benchmarks on Haswell and POWER 8.

We invite external analysis and verification of Randen's properties and
suggest it as a safer alternative to arguably obsolete \cite{pcg}[p.~6]
algorithms such as small linear congruential generators, linear feedback shift
registers, well equidistributed long-period linear \cite{well}, unaugmented
XorShift, and MT.

{
\small
\bibliographystyle{unsrtnat}
\bibliography{references}{}
}

\appendix
\section{Active Functions in 16-branch Feistel}
\label{sec:search}

\begin{lemma}
A type-2 generalized Feistel network with 16 branches and an improved block
shuffle \cite{improvedFeistel} has at least as many differentially active
functions as listed in Table~\ref{tab:active}.
\end{lemma}

\noindent
To the best of our knowledge, these bounds are new. Note that the 6-round
bound is the same as reported for a type-2 network with four branches
\cite{doubleSP}. We will establish our bounds via exhaustive enumeration.
Type-2 Feistel networks update their odd branches by XORing them with the
result of a function of the corresponding even branch:
\verb|new_odd := F(even) XOR odd|. There are two simple properties (numbered
3 and 4 \cite{boundDoubleSP}[p.~83]) regarding the propagation of differences.
First, if both \texttt{even} and \texttt{odd} are differentially inactive,
then so is \verb|new_odd|. Second, at least two of them are active if any of
the three are active. Thus, given an input configuration (i.e.\ whether each
branch is differentially active), the output is active if exactly one input is
active. The inputs are booleans, so this corresponds to simply XORing them.
Next, our simulator counts the number of active functions (i.e.\ the number of
differentially active even input branches), shuffles the outputs and passes
them as inputs to the next round. This process is repeated for every round up
to the desired limit. We consider all $2^{16}$ input configurations except the
trivial case of zero input differences. This logic is implemented by the
following Python script.

\definecolor{green}{rgb}{0,0.6,0}
\begin{lstlisting}[frame=single,numbers=none,language=Python,
commentstyle=\color{green},keywordstyle=\color{blue}]
# (Over)estimates a lower bound of differentially active
# functions in a 16-branch generalized Feistel network.

ROUNDS = range(6)
BRANCHES = 16

idx = range(BRANCHES / 2)    # indices within odd/even
bit_shifts = [2 * i for i in idx]

# Shuffle: `Improving the Generalized Feistel' No.10
shuffle_for_new_odd  = (3,6,5,1,7,4,0,2) # = even i / 2
shuffle_for_new_even = (1,2,4,3,0,5,7,6) # = odd i / 2

min_active_funcs = 99999

def XorResult(even, xor):
	# Page 83 in `Generalized Feistel networks revisited'.
	# 3) if even (input to F) and the XOR input are both
	#    zero (inactive), so is the XOR result.
	if even == 0 and xor == 0: return 0
	# 4) otherwise, at least two of the inputs/output are
	#    active => an inactive input implies active output.
	if even == 0 and xor == 1: return 1
	if even == 1 and xor == 0: return 1
	# Assume inactive => overestimate the lower bound!
	return 0

# For every combination of differentially active
# branches except all-zero (no active functions):
for bits in range(1, 1 << BRANCHES):
	# Extract bits into integers, partition into even/odd.
	even = [((bits >> bit_shifts[i]) & 1) for i in idx]
	odd = [((bits >> (bit_shifts[i] + 1)) & 1) for i in idx]

	# Total differentially active functions.
	active_funcs = 0
	for round in ROUNDS:
		# Active functions (nonzero even[]) in this round.
		active_funcs += even.count(1)

		# Shuffle(even) will later replace the current odd.
		new_odd = [even[shuffle_for_new_odd[i]] for i in idx]
		# Shuffle(F(even, odd)) replaces the current even.
		f_out = [XorResult(even[i], odd[i]) for i in idx]
		even = [f_out[shuffle_for_new_even[i]] for i in idx]
		odd = new_odd

	# Remember and report the lowest.
	min_active_funcs = min(min_active_funcs, active_funcs)

print min_active_funcs
\end{lstlisting}

Note an important limitation of this algorithm: Property~4 does not provide
any guidance when both inputs are active. The differences may cancel, or not.
Thus, this algorithm does not guarantee a lower bound, but it does indicate
such a bound is at most six. We now extend the search
to cover all these possibilities and thus obtain a lower bound.
The search can be paused and resumed from a `state' consisting of the round
number, odd/even status, and the number of active functions so far. When both
inputs are active, we enqueue new states with every possible combination of
the output. Although quick to compute for six rounds, additional rounds yield
trillions of possible combinations. We retain a brute-force approach, but add
some optimizations to make the search tractable. First, the odd and even
differentially-active status can be represented as separate bit arrays, such
that all calls to \texttt{XorResult} simplify to a single 8-bit XOR and the
shuffle reduces to an 8-bit table lookup. Second, we can prune search areas
where \verb|active_funcs| already exceeds the minimum seen so far, because
they will not influence the lower bound. Third, a fixed-size priority-queue
with bitwise operations reduces the space and time overhead to constants. The
C++ source code corresponding to this description will later be open-sourced
alongside the Randen implementation \cite{randenGithub}. It can trace about a
trillion combinations arising during 18 rounds within a few minutes on a
workstation with 24 cores.

\end{document}